\documentclass[sigconf]{acmart}

\usepackage{tikz}
\usetikzlibrary{arrows.meta,positioning,calc}
\usepackage{mathtools}
\usepackage{cleveref}

\theoremstyle{acmplain}
\newtheorem{prop}{Proposition}[section]
\crefname{prop}{Proposition}{Propositions}

\newcommand{\R}{\mathbb{R}}
\newcommand{\I}{\mathcal{I}}

\newcommand{\Static}{\textsc{Static}}
\newcommand{\PlanOne}{\textsc{Plan-1}}
\newcommand{\PlanK}{\textsc{Plan-}K}
\newcommand{\vkip}{VK\_IP}
\newcommand{\vkup}{VK\_UP}

\acmConference[KDD CJ Workshop '26]{5th Workshop on End-to-End Customer Journey Optimization at KDD 2026}{August 9, 2026}{Jeju Island, Republic of Korea}
\acmBooktitle{5th Workshop on End-to-End Customer Journey Optimization at KDD 2026, August 09, 2026, Jeju Island, Republic of Korea}
\acmYear{2026}
\acmISBN{}
\acmDOI{}

\title[Planning over Matrix-Factorization MDPs for Candidate Generation]{Planning over Matrix-Factorization MDPs\\
for Candidate Generation}

\author{Mikhail Trapeznikov}
\affiliation{
  \institution{AI VK}
  \city{Moscow}
  \country{Russia}}
\affiliation{
  \institution{Lomonosov Moscow State University}
  \city{Moscow}
  \country{Russia}}
\email{trapeznikovmy@my.msu.ru}

\author{Maksim Utushkin}
\affiliation{
  \institution{AI VK}
  \city{Moscow}
  \country{Russia}}
\email{mak.utushkin@gmail.com}

\begin{document}

\begin{abstract}
For a recommender service, we view the customer journey as a chain of item recommendations: a useful item changes the user's state and therefore what should be retrieved next. Standard matrix-factorization retrieval ignores this---it builds one user vector and returns the top-$K$ items by a static score, treating them as independent. We ask a narrow question: when is it worth planning over the user-state dynamics that fold-in induces? To answer it we propose casting top-$K$ retrieval as an MDP over the implicit-ALS posterior $(A^{-1},u)$, where an action is an item and the transition is a closed-form rank-one fold-in, and the trajectory reward combines a relevance similarity with a posterior-alignment term. Under the same fixed embeddings we compare static retrieval, one-step planning, and horizon-$K$ MCTS across five datasets and two protocols: a per-user leave-last-$n$ split and a stricter global time split. Dynamics-aware planning tends to overcome static retrieval on all datasets under leave-last-$n$, and the gains hold on MovieLens-1M and the VK-LSVD slices under the global time split. A single step of lookahead already captures most of the gain, so the lightweight planning layer turns static top-$K$ scoring into a short decision and improves retrieval over fixed collaborative-filtering embeddings, with no retraining and no change to the representation. These gains depend on measuring relevance with cosine rather than inner-product similarity, which is otherwise entangled with item popularity.
\end{abstract}



\maketitle

\section{Introduction}
Recommender systems help users navigate large catalogs of products, videos, music and
other items, and so they shape the customer experience on most digital platforms. A
single recommendation is rarely an individual event: a good pick moves the user further
along their journey, a poor one ends the session or pushes them into an irrelevant part
of the catalog. So a recommendation always comes with two effects -- the immediate
choice, and the way it changes what the user will value next.

Modern systems take this seriously by modelling the user history as a sequence of
interactions and feeding it through learned sequence encoders -- RNNs and, more
recently, transformer-based architectures \citep{hidasi2016session,kang2018sasrec,sun2019bert4rec}.
The underlying view is an MDP \citep{shani2005mdp,ie2019slateq}: the state describes the user's current
context, the action is the next recommended item (or slate), the reward is the user's
feedback or a surrogate quality metric, and the transition encodes how the state evolves
after the interaction. Our setting is deliberately different. We do not learn a sequence
encoder and we do not compare against one: a static MF user vector carries no order
information, so a transformer recommender such as BERT4Rec is solving a different problem
on a different input, and dropping it in as a baseline would confound the representation
with the planning layer we want to isolate. We keep the representation fixed and ask only
what a thin decision layer adds on top of it.

In parallel, matrix factorization \citep{koren2009mf,hu2008implicit} did not
go away. Implicit-feedback ALS still anchors many production retrieval stacks: it is
simple, trains and updates cheaply, yields compact user/item embeddings, and integrates
with ANN indices for milliseconds candidate retrieval at catalog scale
\cite{rendle2022revisiting}. The natural question is therefore but whether one can add a small amount of decision-making
\emph{on top of} MF without giving up the embedding-retrieval approach.

The gap we target is concrete. Standard MF retrieval is static: it builds a single user
vector $u$ and returns the top-$K$ items by $\langle u,v_i\rangle$, treating those $K$
items \emph{independently}. But a top-$K$ list \emph{is} a sequence -- if the first item
is consumed, the user's iALS posterior changes, and the second item ought to be picked
relative to the updated state. Static retrieval cannot distinguish between ``$K$
independently good items'' and ``a good ordered trajectory of $K$ items''.

Contextual bandits such as LinUCB also add a decision layer over fixed
embeddings, maintaining a per-user ridge posterior $(A^{-1},u)$ and selecting
items by an upper-confidence index \citep{li2010contextual}. They treat the user as
a \emph{stationary} target, so the optimal policy is one-step and therefore look-ahead only
serves exploration. We instead read the fold-in as a state transition that
changes what to show next, and plan over it for an arbitrary horizon. This makes
retrieval order-dependent: showing $a$ then $b$ yields a different
trajectory reward than $b$ then $a$, since each item is scored against the state
its predecessor produced (Figure~\ref{fig:example}).

\textbf{Contributions.}
(i)~We propose an MDP built directly on top of implicit ALS, in which the state is the
iALS posterior $[A^{-1},u]$, the action is an item, and the transition is a closed-form
Sherman--Morrison fold-in.
(ii)~We pair it with a trajectory reward that combines a relevance similarity with an
\emph{alignment} term grounded in a closed-form characterisation of the posterior shift.
(iii)~We use this MDP to span a depth-spectrum of retrievers --
\textsc{static} (no look-ahead), \textsc{Plan-1} (one-step dynamics-aware re-ranking)
and \textsc{Plan-}$K$ (MCTS, horizon $K$) -- all sharing the \emph{same} fixed item
embeddings, so the comparison isolates the incremental effect of planning from any
representation-learning difference.

\section{Matrix-Factorization MDP}
\label{sec:mfmdp}

\subsection{Implicit-ALS posterior as state}
Let $\I$ be the item catalog and let each item have a fixed embedding
$v_i\in\R^d$ learned by implicit ALS. Following the implicit-feedback
convention, each observed interaction $r_i$ of the current user is split into a
binary preference $p_i=1[r_i>0]$ and a confidence $c_i$ that grows with
$r_i$ (e.g.\ $c_i=1+\alpha r_i$); the graded signal enters through $c_i$, not
through the target. For a user history $H\subset\I$ the fold-in solves
\begin{equation}
  u(H)=\arg\min_{u\in\R^d}\sum_{i\in \I}c_i\,(p_i-u^\top v_i)^2+\lambda\lVert u\rVert_2^2,
  \label{eq:u_H}
\end{equation}
and, since $p_i=1$ on $H$ and $p_i=0$ otherwise,
\begin{equation}
  A(H)=\lambda I+\sum_{i\in \I} c_i\, v_i v_i^\top,\quad u(H)=A(H)^{-1} \left( \sum_{i\in H}c_i\, v_i \right).
  \label{eq:ials}
\end{equation}

We use
\begin{equation}
  s(H)=\big[P(H),u(H)\big],
  \label{eq:s_H_P_H}
\end{equation}
as the MDP state, where $P(H)=A(H)^{-1}$. Keeping $P$ as part of the state is important: it contains the local posterior geometry of the user's history and lets us update a hypothetical interaction without re-solving the least-squares problem.

The available actions are unconsumed items $a\in\I\setminus H$. In the planner we use an optimistic transition: after choosing $a$, the item is treated as accepted and added to the positive history. This is not a claim that every served item will be accepted by a real user. It is a model-based approximation used to construct an ordered candidate trajectory.

\subsection{Closed-form fold-in transition}
For a current state $(P,u)$ and action $a$, write $v=v_a$ and
\begin{equation}
  z=Pv,\qquad \ell=v^\top Pv=v^\top z.
  \label{eq:z_ell}
\end{equation}
The posterior precision receives a rank-one update, $A^+=A+vv^\top$. By Sherman-Morrison (Appendix~\ref{app:sm-foldin}),
\begin{equation}
  P^+=P-\frac{z z^\top}{1+\ell}.
  \label{eq:sm}
\end{equation}
Since $b^+=b+v$, the updated user vector is
\begin{equation}
  u^+=u+\frac{z(1-u^\top v)}{1+\ell}.
  \label{eq:uupdate}
\end{equation}
Thus each hypothetical transition costs $O(d^2)$ rather than $O(d^3)$. The score change of the selected item has the closed form
\begin{equation}
  \Delta_s(a)=(u^+)^\top v-u^\top v
  =\frac{\ell}{1+\ell}(1-u^\top v).
  \label{eq:delta}
\end{equation}
The leverage $\ell$ controls how much the candidate can move the posterior: two items with similar static scores may have different transition effects.

\subsection{Reward}
For a trajectory $\tau=(a_0,\ldots,a_{K-1})$ and states $s_k=[P_k,u_k]$ produced by the fold-in transition, we evaluate
\begin{equation}
  R(\tau)=\sum_{k=0}^{K-1}\gamma^k\,r(s_k,a_k),\quad
  r(s,a)=\mathrm{sim}(u,v_a)+\eta\frac{1}{1+|\Delta_s(a)|}.
  \label{eq:reward}
\end{equation}
The first term, $\mathrm{sim}(u,v_a)$, is the relevance of item $a$ under the current state, where $\mathrm{sim}$ is a similarity between the user and item embeddings. We consider two instantiations, the raw inner product $\mathrm{dot}(u,v_a)=u^\top v_a$ and the cosine $\mathrm{cos}(u,v_a)=u^\top v_a/(\lVert u\rVert\,\lVert v_a\rVert)$, and we compare them empirically (Section~\ref{sec:validation}, Table~\ref{tab:ablation}). The second term is a dynamic-based alignment score: it is large when folding in the item does not sharply change its own score, i.e., when the current posterior is already geometrically consistent with that item. This is not meant to be a direct model of user satisfaction; it is a planning value defined over fixed MF embeddings.

\section{Planning for Top-$K$ Retrieval}
\label{sec:planning}

\subsection{From a static slate to a trajectory}
A retrieval system normally returns a set or ordered list of $K$ candidates. Static matrix-factorization retrieval scores all candidates from the root state $s_0$ and takes the top $K$. In the MDP view, the same output can be treated as a trajectory:
\begin{equation}
  \tau_K^*=\arg\max_{a_0,\ldots,a_{K-1}}
  \sum_{k=0}^{K-1}\gamma^k r(s_k,a_k),\qquad
  s_{k+1}=T(s_k,a_k).
  \label{eq:traj}
\end{equation}
This objective distinguishes a set of individually good items from a sequence whose early choices move the posterior in a useful direction.

We compare three depths. \Static{} is depth zero: all candidates are scored at $s_0$. \PlanOne{} (see Fig.~\ref{fig:plan1-vs-plank}a) evaluates each candidate after a single fold-in and re-ranks by the one-step relevance/alignment value. \PlanK{} (see Fig.~\ref{fig:plan1-vs-plank}b) runs tree search against the discounted trajectory objective over a longer horizon.

\begin{figure}[t]
  \centering
  \includegraphics[width=\columnwidth]{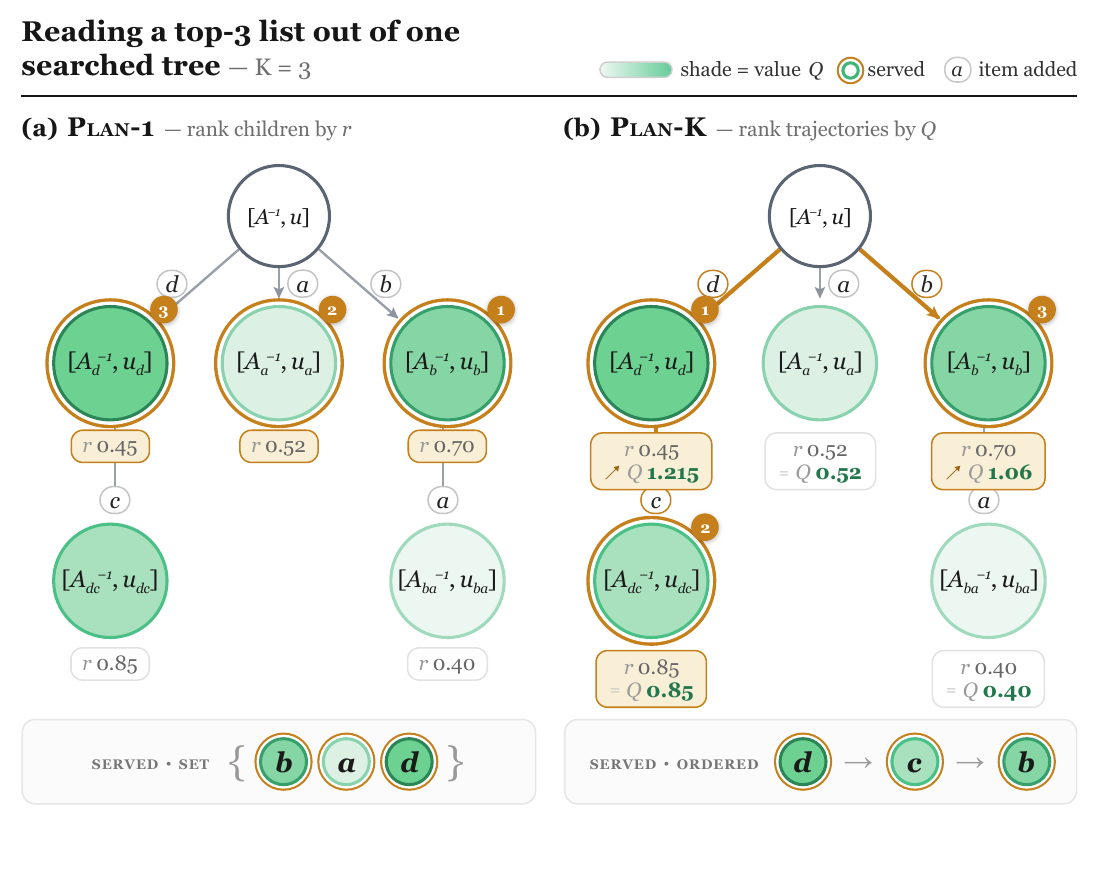}
  \caption{\textbf{Two read-outs of the same searched tree ($K=3$).}
Both panels use the same iALS posterior tree: edges are folded-in items and node shade denotes backed-up value $Q$.
\textsc{Plan-1} ranks root children by immediate score and returns the best items independently, while \textsc{Plan-K} ranks trajectories by discounted return.}
  \Description{Two panels compare Plan-1 and Plan-K read-outs from the same searched tree. Plan-1 selects root children by immediate score, while Plan-K selects backed-up trajectories and can choose an item that unlocks a stronger continuation.}
  \label{fig:plan1-vs-plank}
\end{figure}

\subsection{MCTS in the optimistic environment}
The branching factor is the catalog size, so exact search over \eqref{eq:traj} is infeasible. We use MCTS with two retrieval-specific choices. First, expansion is limited to the $l$ nearest items to the current user vector according to an ANN index; this is the same idea as Wolpertinger-style action reduction in large discrete spaces \citep{dulacarnold2015wolpertinger}. Second, the leaf value is computed in closed form using the fold-in equations rather than by a learned simulator.

At a visited node $s$, selection uses a UCT-like score \citep{kocsis2006uct}
\begin{equation}
  \pi(s,a)=Q(s,a)+c_{\mathrm{uct}}
  \sqrt{\frac{\log N(s)\,[2+\cos(u_s,v_a)]}{N(s,a)+1}}.
  \label{eq:uct}
\end{equation}
The exploration prior uses cosine because it is bounded ($2+\cos\in[1,3]$), independently of the relevance similarity $\mathrm{sim}$ chosen in \eqref{eq:reward}; the corresponding term stays positive and keeps the usual exploration behavior (Appendix~\ref{app:uct-regret}). Backpropagation uses discounted return-to-go, $G\leftarrow r(s_j,a_j)+\gamma G$, so the tree estimates the trajectory reward rather than a frontier re-ranking score (see Fig.~\ref{fig:mcts}).

\begin{figure}[t]
  \centering
  \includegraphics[width=\columnwidth]{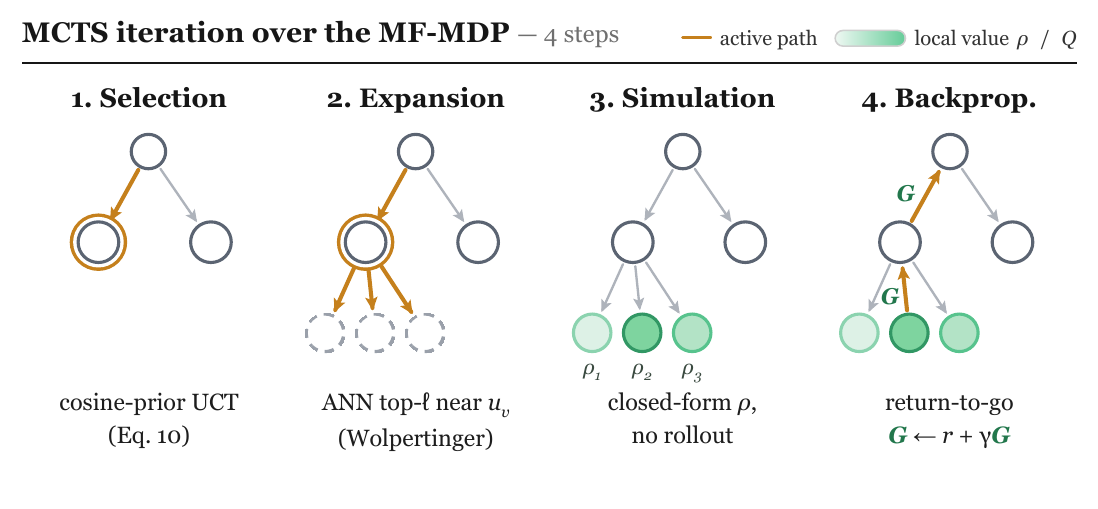}
  \caption{One MCTS iteration over the MF-MDP: selection by cosine-prior UCT, ANN-based expansion, closed-form leaf evaluation, and discounted return-to-go backpropagation.}
  \Description{Four panels show selection, expansion, simulation, and backpropagation in one MCTS iteration over the matrix-factorization MDP.}
  \label{fig:mcts}
\end{figure}

\section{Validation}
\label{sec:validation}

\subsection{Datasets and protocols}
We evaluate the planner on five domains. MovieLens-1M uses movie ratings thresholded at $4$ stars \citep{harper2015movielens}. KuaiRec is a short-video dataset with positive feedback defined by high watch ratio \citep{gao2022kuairec}. VK-LSVD is a large industrial short-video dataset; we use popular-item and popular-user subsamples, \vkip{} and \vkup{} \citep{poslavsky2026vklsvd}. YAMBDA is a music-domain benchmark; in our experiments, a positive event is a listen with played ratio at least $80\%$ \citep{ploshkin2025yambda}. Large VK/YAMBDA catalogs are restricted to a top-5000 popular candidate pool; MovieLens and KuaiRec use their native catalogs.

All methods use the same item factors trained once on the training split of the corresponding protocol. The compared variants differ only in how they construct the top-$K$ list: \Static{} scores from the root state, \PlanOne{} applies the proposed one-step relevance/alignment value after a hypothetical fold-in, and \PlanK{} is the MCTS planner with the same reward. This design keeps representation learning fixed and makes the reported comparison about the evidence added by a planning layer. Because inner-product scores over collaborative-filtering embeddings are entangled with item popularity~\citep{steck2011popularity}, the choice of relevance similarity matters; unless stated otherwise the planners use the cosine instantiation of $\mathrm{sim}$, and we isolate the effect of the dot/cos choice on the strongest planner as an ablation (Table~\ref{tab:ablation}).

We report two protocols. \emph{Leave-last-$n$} (LLN) holds out the last interactions of each user and uses the earlier prefix to build the initial posterior. It tests whether the planner can continue a user's observed trajectory. \emph{Global time split} (GTS) separates train, validation, and test by global timestamp cutoffs, so evaluation interactions occur strictly in the future of the training window. GTS is closer to deployment, but it is also harsher: user activity, item popularity, and catalog composition may drift across the cutoffs.

For both protocols, planner hyperparameters, including $c_{\mathrm{uct}}$, $\gamma$, expansion width, and root readout, are selected on validation only. We report Recall@10 and nDCG@10 on the test split, aggregated by user and seed. Significance is tested against \Static{} using paired Wilcoxon tests.

\begin{table*}[t]
\centering
\caption{Final test results under shared MF embeddings, mean\,$\pm$\,std over $3$ seeds. \Static{} is the conventional inner-product (dot) static retrieval; \PlanOne{} and \PlanK{} use the cosine relevance similarity. Stars mark significant gains versus \Static{} under paired Wilcoxon tests ($^{*}p<0.05$, $^{**}p<0.01$, $^{***}p<0.001$). Boldface marks the best point estimate in each protocol--dataset row for each metric; ties are bolded jointly. The dot/cos similarity choice is isolated in Table~\ref{tab:ablation}.}
\label{tab:main}
\small
\setlength{\tabcolsep}{4.5pt}
\renewcommand{\arraystretch}{1.12}
\begin{tabular}{lcccccc}
\toprule
& \multicolumn{3}{c}{Recall@10} & \multicolumn{3}{c}{nDCG@10} \\
\cmidrule(lr){2-4}\cmidrule(lr){5-7}
Dataset & \Static{} & \PlanOne{} & \PlanK{} & \Static{} & \PlanOne{} & \PlanK{} \\
\midrule
\multicolumn{7}{l}{\textbf{Protocol 1: leave-last-$n$ (LLN)}} \\
ML-1M
& $.0694_{\pm.0032}$ & $.0762_{\pm.0016}$ & $\mathbf{.0764}_{\pm.0015}$
& $.0567_{\pm.0018}$ & $\mathbf{.0612}_{\pm.0017}$ & $.0611_{\pm.0020}$ \\
KuaiRec
& $.0993_{\pm.0105}$ & $.1082_{\pm.0109}$ & $\mathbf{.1107}^{**}_{\pm.0096}$
& $.0712_{\pm.0087}$ & $.0784^{*}_{\pm.0084}$ & $\mathbf{.0800}^{**}_{\pm.0078}$ \\
\vkip{}
& $.0201_{\pm.0033}$ & $\mathbf{.0294}^{***}_{\pm.0053}$ & $.0269^{*}_{\pm.0025}$
& $.0152_{\pm.0023}$ & $\mathbf{.0229}^{***}_{\pm.0033}$ & $.0215^{***}_{\pm.0020}$ \\
\vkup{}
& $.0161_{\pm.0018}$ & $\mathbf{.0260}^{***}_{\pm.0007}$ & $.0248^{***}_{\pm.0024}$
& $.0123_{\pm.0014}$ & $\mathbf{.0192}^{***}_{\pm.0013}$ & $.0187^{***}_{\pm.0024}$ \\
YAMBDA
& $.0138_{\pm.0015}$ & $.0177_{\pm.0035}$ & $\mathbf{.0179}^{*}_{\pm.0036}$
& $.0112_{\pm.0008}$ & $.0147_{\pm.0033}$ & $\mathbf{.0148}_{\pm.0035}$ \\
\midrule
\multicolumn{7}{l}{\textbf{Protocol 2: global time split (GTS)}} \\
ML-1M
& $.0429_{\pm.0028}$ & $\mathbf{.0453}_{\pm.0038}$ & $.0429_{\pm.0037}$
& $.1446_{\pm.0126}$ & $\mathbf{.1552}^{*}_{\pm.0083}$ & $.1484_{\pm.0106}$ \\
KuaiRec
& $\mathbf{.0321}_{\pm.0054}$ & $.0286_{\pm.0079}$ & $.0288_{\pm.0084}$
& $.0324_{\pm.0069}$ & $.0322_{\pm.0079}$ & $\mathbf{.0328}_{\pm.0086}$ \\
\vkip{}
& $.0093_{\pm.0027}$ & $\mathbf{.0108}_{\pm.0031}$ & $\mathbf{.0108}_{\pm.0031}$
& $.0080_{\pm.0014}$ & $\mathbf{.0106}_{\pm.0019}$ & $\mathbf{.0106}_{\pm.0019}$ \\
\vkup{}
& $.0115_{\pm.0024}$ & $\mathbf{.0162}^{**}_{\pm.0048}$ & $\mathbf{.0162}^{**}_{\pm.0048}$
& $.0115_{\pm.0024}$ & $\mathbf{.0160}^{***}_{\pm.0024}$ & $\mathbf{.0160}^{***}_{\pm.0024}$ \\
YAMBDA
& $\mathbf{.0306}_{\pm.0018}$ & $.0256_{\pm.0017}$ & $.0256_{\pm.0017}$
& $.0623_{\pm.0020}$ & $\mathbf{.0628}_{\pm.0016}$ & $\mathbf{.0628}_{\pm.0015}$ \\
\bottomrule
\end{tabular}
\end{table*}

\begin{table}[t]
\centering
\caption{Ablation on the relevance similarity in the most consistent planner \PlanOne{}: inner product (dot) vs cosine (cos), mean\,$\pm$\,std over $3$ seeds. Boldface marks the better variant per row and metric.}
\label{tab:ablation}
\small
\setlength{\tabcolsep}{4.5pt}
\renewcommand{\arraystretch}{1.12}
\begin{tabular}{lcccc}
\toprule
& \multicolumn{2}{c}{Recall@10} & \multicolumn{2}{c}{nDCG@10} \\
\cmidrule(lr){2-3}\cmidrule(lr){4-5}
Dataset & \PlanOne{}(dot) & \PlanOne{}(cos) & \PlanOne{}(dot) & \PlanOne{}(cos) \\
\midrule
\multicolumn{5}{l}{\textbf{Protocol 1: leave-last-$n$ (LLN)}} \\
ML-1M
& $.0630_{\pm.0013}$ & $\mathbf{.0762}_{\pm.0016}$
& $.0513_{\pm.0019}$ & $\mathbf{.0612}_{\pm.0017}$ \\
KuaiRec
& $.0910_{\pm.0117}$ & $\mathbf{.1082}_{\pm.0109}$
& $.0650_{\pm.0090}$ & $\mathbf{.0784}_{\pm.0084}$ \\
\vkip{}
& $.0146_{\pm.0022}$ & $\mathbf{.0294}_{\pm.0053}$
& $.0104_{\pm.0016}$ & $\mathbf{.0229}_{\pm.0033}$ \\
\vkup{}
& $.0127_{\pm.0010}$ & $\mathbf{.0260}_{\pm.0007}$
& $.0095_{\pm.0006}$ & $\mathbf{.0192}_{\pm.0013}$ \\
YAMBDA
& $.0132_{\pm.0011}$ & $\mathbf{.0177}_{\pm.0035}$
& $.0106_{\pm.0003}$ & $\mathbf{.0147}_{\pm.0033}$ \\
\midrule
\multicolumn{5}{l}{\textbf{Protocol 2: global time split (GTS)}} \\
ML-1M
& $.0436_{\pm.0045}$ & $\mathbf{.0453}_{\pm.0038}$
& $.1345_{\pm.0109}$ & $\mathbf{.1552}_{\pm.0083}$ \\
KuaiRec
& $\mathbf{.0300}_{\pm.0046}$ & $.0286_{\pm.0079}$
& $.0300_{\pm.0063}$ & $\mathbf{.0322}_{\pm.0079}$ \\
\vkip{}
& $.0089_{\pm.0014}$ & $\mathbf{.0108}_{\pm.0031}$
& $.0081_{\pm.0010}$ & $\mathbf{.0106}_{\pm.0019}$ \\
\vkup{}
& $.0096_{\pm.0014}$ & $\mathbf{.0162}_{\pm.0048}$
& $.0105_{\pm.0016}$ & $\mathbf{.0160}_{\pm.0024}$ \\
YAMBDA
& $\mathbf{.0301}_{\pm.0022}$ & $.0256_{\pm.0017}$
& $.0609_{\pm.0013}$ & $\mathbf{.0628}_{\pm.0016}$ \\
\bottomrule
\end{tabular}
\end{table}

\begin{figure*}[t]
\centering
\includegraphics[width=0.92\textwidth]{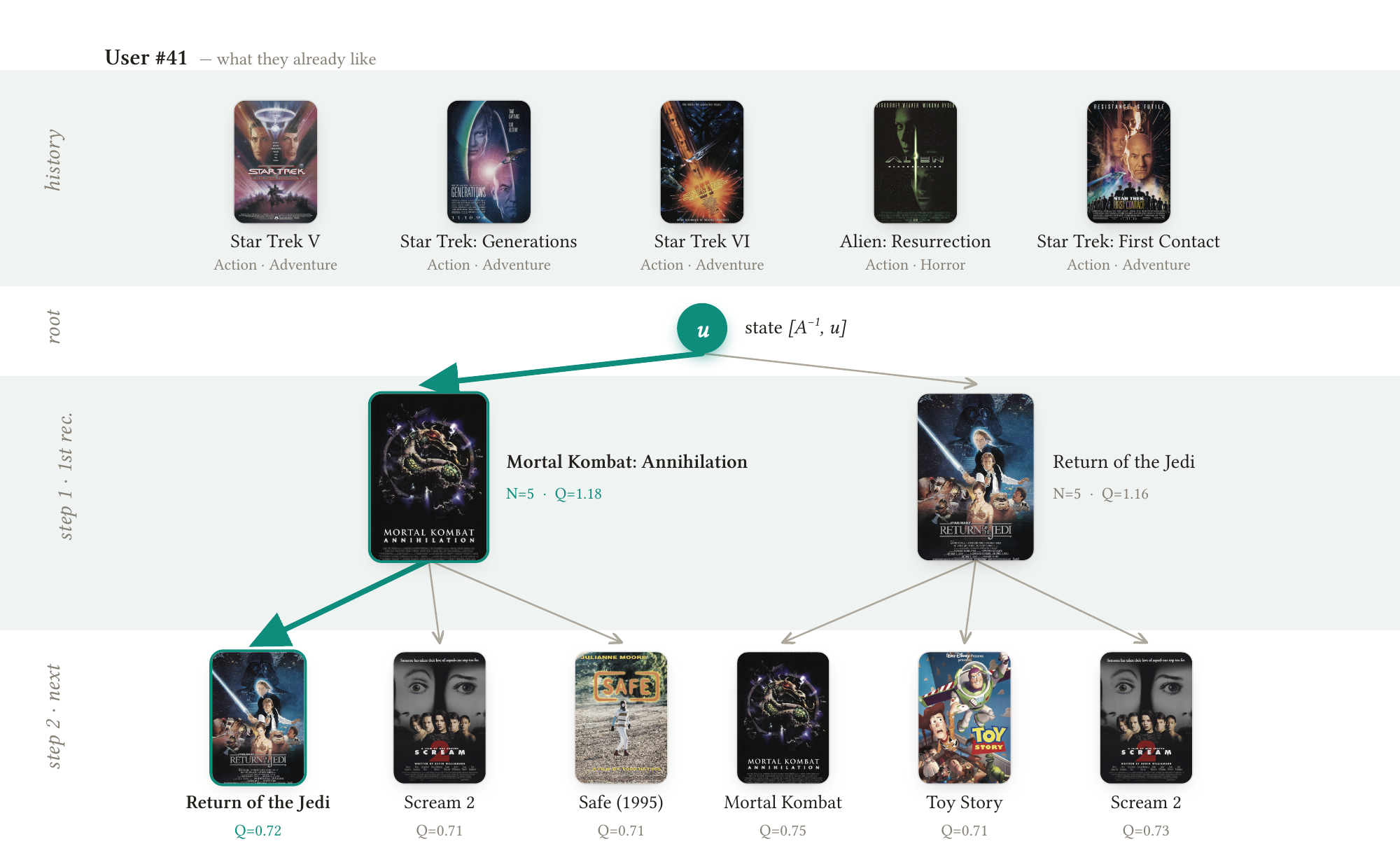}
\caption{A MovieLens example of the tree-search planner. The top row shows a user's consumed history; the root is the posterior state $(A^{-1},u)$; the first expanded level contains candidate first recommendations with visit/value annotations; the second level shows the planned continuation. The highlighted path is the selected trajectory.}
\Description{A recommendation tree for one movie user, with a highlighted selected trajectory and other explored branches.}
\label{fig:example}
\end{figure*}

\subsection{Results}
Table~\ref{tab:main} gives the main comparison under cosine relevance. The LLN panel is where the MDP formulation shows the clearest signal: \PlanOne{} improves Recall@10 over \Static{} on all five datasets, significantly so on both VK-LSVD slices, while \PlanK{} attains the best recall on MovieLens-1M, KuaiRec and YAMBDA and is significant on KuaiRec (Recall and nDCG) and YAMBDA (Recall). The two VK-LSVD slices behave the same way: one-step planning is already the best variant and deeper search adds nothing useful on top of it.

Under the GTS protocol, the one-step value remains significant on both VK-LSVD subsamples and on MovieLens-1M nDCG, but \Static{} takes the best recall on KuaiRec and YAMBDA, where planning loses outright. We read this as temporal drift the optimistic fold-in does not model: when item popularity and catalog composition move between the training window and the test window, an accepted-and-folded-in item is a worse description of the next state than it is under LLN.

The relevance similarity is what makes any of this work, so we isolate it on the most consistent planner, \PlanOne{} (Table~\ref{tab:ablation}). The dot-product instantiation amplifies popular items along the trajectory; the cosine one drops the norm and keeps the relevance term popularity-robust, and it is cosine that produces every gain reported above. The only cells where dot is not the worse choice are the two GTS feed domains (KuaiRec, YAMBDA), and there neither similarity beats static and popularity itself has already drifted.

Deep MCTS earns its cost only sometimes. \PlanK{} beats \PlanOne{} under LLN on KuaiRec and gives marginal recall gains on MovieLens-1M and YAMBDA, but it does not move the GTS numbers and trails the one-step method on the VK slices. So the order we would actually deploy is: \PlanOne{} under cosine relevance first, and \PlanK{} only where validation shows sequential structure worth searching for.

\section{Conclusion and Future Work}
We introduced a simple MDP on top of implicit-ALS embeddings. The user state is the posterior pair $[A^{-1},u]$, the transition is defined by a rank-one update in a closed form, and the reward combines current relevance with posterior alignment. Because the item representation is kept fixed, the experiments provide evidence about the incremental value of adding a planning layer to the same MF retrieval geometry, rather than about a new representation model.

Under leave-last-$n$, dynamics-aware planning outperforms static retrieval on all five datasets. Under the global time split the gain remains on MovieLens-1M and the VK-LSVD slices but disappears on KuaiRec and YAMBDA. So the conclusion is not that MCTS should replace retrieval everywhere. A single dynamics-aware fold-in is often a cheap and useful improvement over a static top-$K$ list; deeper search pays off only where validation justifies it, and a temporal split is what tells you whether the gain survives future-period drift in the first place.

A second finding is about the relevance term. The gain depends on cosine similarity: with the raw inner product, planning typically matches or significantly underperforms static retrieval, because inner-product scores amplify item popularity. Cosine is a prerequisite for the effect, not a tuning detail.

The present environment is deliberately optimistic and deterministic: the planner assumes that a selected item is accepted and folded into the positive history. A natural next step is to learn stochastic dynamics and values from logs, including rejection, session termination, and delayed satisfaction. Another direction is to distill the planner into a fast retrieval policy \cite{schrittwieser2020mastering} or to use a sequence model as a value function inside the MF-MDP, connecting this lightweight construction to broader model-based RL for recommendation.

\bibliographystyle{ACM-Reference-Format}
\bibliography{references}

@inproceedings{hu2008implicit,
  title={Collaborative Filtering for Implicit Feedback Datasets},
  author={Hu, Yifan and Koren, Yehuda and Volinsky, Chris},
  booktitle={2008 Eighth IEEE international conference on data mining},
  pages={263--272},
  year={2008},
  organization={Ieee}
}

@article{koren2009mf,
  title={Matrix Factorization Techniques for Recommender Systems},
  author={Koren, Yehuda and Bell, Robert and Volinsky, Chris},
  journal={Computer},
  volume={42},
  number={8},
  pages={30--37},
  year={2009},
  publisher={IEEE}
}

@article{hidasi2016session,
  title={Session-based Recommendations with Recurrent Neural Networks},
  author={Hidasi, Bal{\'a}zs and Karatzoglou, Alexandros and Baltrunas, Linas and Tikk, Domonkos},
  journal={arXiv preprint arXiv:1511.06939},
  year={2015}
}

@inproceedings{kang2018sasrec,
  title={Self-Attentive Sequential Recommendation},
  author={Kang, Wang-Cheng and McAuley, Julian},
  booktitle={2018 IEEE international conference on data mining (ICDM)},
  pages={197--206},
  year={2018},
  organization={IEEE}
}

@inproceedings{sun2019bert4rec,
  title={BERT4Rec: Sequential Recommendation with Bidirectional Encoder Representations from Transformer},
  author={Sun, Fei and Liu, Jun and Wu, Jian and Pei, Changhua and Lin, Xiao and Ou, Wenwu and Jiang, Peng},
  booktitle={Proceedings of the 28th ACM international conference on information and knowledge management},
  pages={1441--1450},
  year={2019}
}

@article{shani2005mdp,
  title={An MDP-based Recommender System},
  author={Shani, Guy and Heckerman, David and Brafman, Ronen I},
  journal={Journal of machine Learning research},
  volume={6},
  number={Sep},
  pages={1265--1295},
  year={2005}
}

@article{ie2019slateq,
  title={Reinforcement Learning for Slate-based Recommender Systems: A Tractable Decomposition and Practical Methodology},
  author={Ie, Eugene and Jain, Vihan and Wang, Jing and Narvekar, Sanmit and Agarwal, Ritesh and Wu, Rui and Cheng, Heng-Tze and Lustman, Morgane and Gatto, Vince and Covington, Paul and others},
  journal={arXiv preprint arXiv:1905.12767},
  year={2019}
}

@article{dulacarnold2015wolpertinger,
  title={Deep Reinforcement Learning in Large Discrete Action Spaces},
  author={Dulac-Arnold, Gabriel and Evans, Richard and van Hasselt, Hado and Sunehag, Peter and Lillicrap, Timothy and Hunt, Jonathan and Mann, Timothy and Weber, Theophane and Degris, Thomas and Coppin, Ben},
  journal={arXiv preprint arXiv:1512.07679},
  year={2015}
}

@inproceedings{kocsis2006uct,
  title={Bandit Based Monte-Carlo Planning},
  author={Kocsis, Levente and Szepesv{\'a}ri, Csaba},
  booktitle={European conference on machine learning},
  pages={282--293},
  year={2006},
  organization={Springer}
}

@inproceedings{steck2011popularity,
  title={Item Popularity and Recommendation Accuracy},
  author={Steck, Harald},
  booktitle={Proceedings of the fifth ACM conference on Recommender systems},
  pages={125--132},
  year={2011}
}

@article{harper2015movielens,
  title={The MovieLens Datasets: History and Context},
  author={Harper, F Maxwell and Konstan, Joseph A},
  journal={Acm transactions on interactive intelligent systems (tiis)},
  volume={5},
  number={4},
  pages={1--19},
  year={2015},
  publisher={Acm New York, NY, USA}
}

@inproceedings{gao2022kuairec,
  title={KuaiRec: A Fully-observed Dataset and Insights for Evaluating Recommender Systems},
  author={Gao, Chongming and Li, Shijun and Lei, Wenqiang and Chen, Jiawei and Li, Biao and Jiang, Peng and He, Xiangnan and Mao, Jiaxin and Chua, Tat-Seng},
  booktitle={Proceedings of the 31st ACM International Conference on Information \& Knowledge Management},
  pages={540--550},
  year={2022}
}

@inproceedings{poslavsky2026vklsvd,
  title={VK-LSVD: A Large-Scale Industrial Dataset for Short-Video Recommendation},
  author={Poslavsky, Aleksandr and D'yakonov, Alexander and Dorn, Yuriy and Zimovnov, Andrey},
  booktitle={Proceedings of the ACM Web Conference 2026},
  pages={8657--8660},
  year={2026},
  series={WWW '26}
}

@inproceedings{ploshkin2025yambda,
  title={Yambda-5B—A Large-Scale Multi-Modal Dataset for Ranking and Retrieval},
  author={Ploshkin, Alexander and Tytskiy, Vladislav and Pismenny, Alexey and Baikalov, Vladimir and Taychinov, Evgeny and Permiakov, Artem and Burlakov, Daniil and Krofto, Eugene},
  booktitle={Proceedings of the Nineteenth ACM Conference on Recommender Systems},
  pages={894--901},
  year={2025}
}

@inproceedings{rendle2022revisiting,
  title={Revisiting the Performance of iALS on Item Recommendation Benchmarks},
  author={Rendle, Steffen and Krichene, Walid and Zhang, Li and Koren, Yehuda},
  booktitle={Proceedings of the 16th ACM Conference on Recommender Systems},
  pages={427--435},
  year={2022}
}

@article{schrittwieser2020mastering,
  title={Mastering Atari, Go, Chess and Shogi by Planning with a Learned Model},
  author={Schrittwieser, Julian and Antonoglou, Ioannis and Hubert, Thomas and Simonyan, Karen and Sifre, Laurent and Schmitt, Simon and Guez, Arthur and Lockhart, Edward and Hassabis, Demis and Graepel, Thore and others},
  journal={Nature},
  volume={588},
  number={7839},
  pages={604--609},
  year={2020},
  publisher={Nature Publishing Group UK London}
}

@article{sherman1950adjustment,
  title={Adjustment of an Inverse Matrix Corresponding to a Change in One Element of a Given Matrix},
  author={Sherman, Jack and Morrison, Winifred J},
  journal={The Annals of Mathematical Statistics},
  volume={21},
  number={1},
  pages={124--127},
  year={1950},
  publisher={JSTOR}
}

@article{lai1985asymptotically,
  title={Asymptotically efficient adaptive allocation rules},
  author={Lai, Tze Leung and Robbins, Herbert},
  journal={Advances in applied mathematics},
  volume={6},
  number={1},
  pages={4--22},
  year={1985},
  publisher={Academic Press, Inc. Orlando, FL, USA}
}

@inproceedings{li2010contextual,
  title={A Contextual-Bandit Approach to Personalized News Article Recommendation},
  author={Li, Lihong and Chu, Wei and Langford, John and Schapire, Robert E},
  booktitle={Proceedings of the 19th international conference on World wide web},
  pages={661--670},
  year={2010}
}

\appendix
\section{ALS Fold-in Update}
\label{app:sm-foldin}
This appendix spells out the rank-one algebra used by the optimistic ALS transition.
All vectors are column vectors.

\begin{prop}[Sherman--Morrison identity~\cite{sherman1950adjustment}]\label{prop:sm}
Let $A\in\mathbb{R}^{d\times d}$ be invertible and let $x,y\in\mathbb{R}^d$ satisfy
$1+y^\top A^{-1}x\ne 0$. Then
\begin{equation}
  (A+xy^\top)^{-1}
  =A^{-1}-\frac{A^{-1}xy^\top A^{-1}}{1+y^\top A^{-1}x}.
  \label{eq:app-sm-identity}
\end{equation}
\end{prop}

Following \cref{eq:u_H,eq:ials,eq:s_H_P_H,eq:z_ell,eq:sm,eq:uupdate}, consider a current ALS fold-in state
\begin{equation}
  A=\lambda I+\sum_{i\in \I} c_i v_i v_i^\top,
  \quad
  b=\sum_{i\in H} c_i\, v_i,
  \quad
  P=A^{-1},
  \quad
  u=Pb .
\end{equation}
Suppose the planner chooses item $a$ with embedding $v=v_a$ and, in the optimistic
planning model, adds it to the positive history. We fold it in as a unit-confidence
positive interaction ($p=1$, $c_a=1$), so the rank-one update keeps its simple form:
\begin{equation}
  A^+=A+vv^\top,
  \quad
  b^+=b+v,
  \quad
  P^+=(A^+)^{-1},
  \quad
  u^+=P^+b^+ .
\end{equation}
Let
\begin{equation}
  z=Pv,
  \qquad
  \ell=v^\top Pv=v^\top z .
\end{equation}

\begin{prop}[closed-form ALS planning update]\label{prop:foldin}
The inverse precision and user vector after the hypothetical interaction are
\begin{align}
  P^+ &= P-\frac{zz^\top}{1+\ell}, \label{eq:app-p-update}\\
  u^+ &= u+\frac{z(1-u^\top v)}{1+\ell}. \label{eq:app-u-update}
\end{align}
Consequently, the self-score shift of the selected item is
\begin{equation}
  \Delta_s(a)=(u^+)^\top v-u^\top v
  =\frac{\ell}{1+\ell}(1-u^\top v).
  \label{eq:app-score-shift}
\end{equation}
\end{prop}

\begin{proof}
Since $\lambda>0$, the matrix $A$ is positive definite and therefore invertible.
The optimistic fold-in adds one positive interaction with item vector $v$, so the
precision matrix receives a rank-one update, $A^+=A+vv^\top$. Applying
\cref{prop:sm} with $x=y=v$ gives
\begin{equation}
  P^+=(A+vv^\top)^{-1}=P-\frac{Pvv^\top P}{1+v^\top Pv}.
\end{equation}
Because $z=Pv$ and $z^\top=v^\top P^\top = v^\top P$, this is exactly \cref{eq:app-p-update}.

For the user vector, expand $u^+=P^+b^+=P^+(b+v)$ and use $u=Pb$, $z=Pv$,
$\ell=v^\top Pv$, together with $v^\top Pb=v^\top u=u^\top v$:
\begin{align}
  u^+
  &= Pb + Pv
     - \frac{Pv\,v^\top Pb}{1+v^\top Pv}
     - \frac{Pv\,v^\top Pv}{1+v^\top Pv} \\
  &= u + z
     - \frac{z\,v^\top u}{1+\ell}
     - \frac{z\,\ell}{1+\ell} = u + z
     - \frac{z\,u^\top v}{1+\ell}
     - \frac{z\,\ell}{1+\ell} \\
  &= u + z\left(1-\frac{u^\top v+\ell}{1+\ell}\right)
   = u + z\,\frac{1-u^\top v}{1+\ell},
\end{align}
which is \cref{eq:app-u-update}. Finally, multiplying $u^+-u$ by $v^\top$ and using
$v^\top z=\ell$,
\begin{equation}
  (u^+)^\top v-u^\top v
  = v^\top(u^+-u)
  = v^\top z\,\frac{1-u^\top v}{1+\ell}
  = \frac{\ell}{1+\ell}(1-u^\top v),
\end{equation}
which proves \cref{eq:app-score-shift}. The computation requires one matrix-vector
product $Pv$ and one rank-one matrix update, so a hypothetical transition costs
$O(d^2)$ and avoids a fresh $O(d^3)$ inversion.
\end{proof}

\section{Cosine-Prior UCT and Node-wise Regret}
\label{app:uct-regret}
The argument below is local to one expanded MCTS node. It should be read as the
standard finite-armed UCB abstraction of the tree-policy step, not as a regret
theorem for the full recommender system or for real user feedback.

\begin{prop}[cosine-prior UCT keeps logarithmic node-wise regret]\label{prop:uct}
Fix a node $s$ with a finite expanded child set $\mathcal{A}_s$. Assume the
normalized return samples observed after choosing child $a\in\mathcal{A}_s$ are
independent, lie in $[0,1]$, and have mean $\mu_a$. Let $\widehat\mu_a(t)$ be the
empirical mean after $N_a(t)$ visits, let $a^\star\in\arg\max_a\mu_a$, and define
$\Delta_a=\mu_{a^\star}-\mu_a$ for each suboptimal child. Consider the tree policy
\begin{equation}
  a_t\in\arg\max_{a\in\mathcal{A}_s}
  \widehat\mu_a(t)
  + c_{\mathrm{uct}}
  \sqrt{\frac{\alpha_a\log t}{N_a(t)}},
  \qquad
  \alpha_a = 2+\cos(u_s,v_a),
  \label{eq:app-cos-uct}
\end{equation}
with $c_{\mathrm{uct}}\ge 1$, after each child has been visited once. Then the
node-wise regret retains its logarithmic order.
\end{prop}

\begin{proof}
First we bound the cosine multiplier. For nonzero vectors $u_s,v_a$,
\begin{equation}
  \cos(u_s,v_a)=\frac{u_s^\top v_a}{\lVert u_s\rVert_2\lVert v_a\rVert_2},
\end{equation}
and by Cauchy--Schwarz $|u_s^\top v_a|\le\lVert u_s\rVert_2\lVert v_a\rVert_2$, hence
$-1\le\cos(u_s,v_a)\le 1$ and
\begin{equation}
  1\le \alpha_a=2+\cos(u_s,v_a)\le 3 .
  \label{eq:app-alpha-range}
\end{equation}
This is the only property of the cosine prior used: it is deterministic at node $s$,
positive, and uniformly bounded. If an implementation encounters a zero vector and
sets the cosine to $0$, the bound still holds.

Define the exploration radius $b_a(t,n)=c_{\mathrm{uct}}\sqrt{\alpha_a\log t/n}$. The
Chernoff--Hoeffding inequality for variables in $[0,1]$ gives, for $n=N_a(t)$,
\begin{align}
  \Pr\{\mu_a\ge \widehat\mu_a(t)+b_a(t,n)\}
  &\le \exp\{-2n\, b_a(t,n)^2\} \\
  &= \exp\{-2c_{\mathrm{uct}}^2\alpha_a\log t\}
   = t^{-2c_{\mathrm{uct}}^2\alpha_a}
  \le t^{-2},
  \label{eq:app-hoeffding-lower}
\end{align}
the last step by $c_{\mathrm{uct}}\ge1$ and $\alpha_a\ge1$. The symmetric
upper-deviation event obeys the same bound,
\begin{equation}
  \Pr\{\widehat\mu_a(t)\ge \mu_a+b_a(t,n)\}
  \le t^{-2c_{\mathrm{uct}}^2\alpha_a}
  \le t^{-2}.
  \label{eq:app-hoeffding-upper}
\end{equation}
Thus the cosine-prior bonus preserves the summable confidence-failure probabilities
of the standard UCB argument~\cite{lai1985asymptotically}, so node-wise regret keeps
its logarithmic order.
\end{proof}

\end{document}